\documentclass{LMCS}

\usepackage{amsmath}
\usepackage{amssymb}

\usepackage{gastex,enumerate,hyperref}
\renewcommand{\emptyset}{\varnothing}

\def\doi{4 (1:8) 2008}
\lmcsheading%
{\doi}
{1--23}
{}
{}
{Jul.~21, 2006}
{Mar.~25, 2008}
{}   

\begin{document}

\title[Feferman-Vaught theorem, automata and logic]{An Application of the Feferman-Vaught Theorem to Automata and Logics for Words over an Infinite Alphabet}

\author{Alexis B\`es}	%required
\address{Laboratoire d'Algorithmique, Complexit\'e et Logique, EA 4213, Universit\'e Paris-Est}	%required
\email{bes@univ-paris12.fr}  %optional

%% mandatory lists of keywords and classifications:
\keywords{Feferman-Vaught method, composition theorems, decidability, automata, infinite alphabet}
\subjclass{F.1.1, F.4.1, F.4.3}

%%%%%%%%%%%%%%%%%%%%%%%%%%%%%%%%%%%%%%%%%%%%%%%%%%%%%%%%%%%%%%%%%%%%%%%%%%%

%% the abstract has to PRECEED the command \maketitle:
%% be sure not to issue the \maketitle command twice!

\begin{abstract}
  \noindent We show that a special case of the Feferman-Vaught composition theorem gives rise to a natural notion of automata for finite words over an infinite alphabet, with good closure and decidability properties, as well as several logical characterizations.  
We also consider a slight extension of the Feferman-Vaught formalism which allows to express more relations between component values (such as equality), and prove related decidability results. From this result we get new classes of decidable logics for words over an infinite alphabet.

\end{abstract}

\maketitle

\section*{Introduction}

The problem of finding suitable notions of automata for words over an infinite alphabet has been adressed in several papers \cite{ABB80,KamFra94,BDMSS05,BPT03,CG05a,NSV04}. The motivations are e.g. modelization of temporized systems, distributed systems, or manipulation of semi-structured data. A common goal is to find a simple and expressive model which preserves as much as possible the good properties of the classical model. Kaminski and Francez \cite{KamFra94} introduce {\em finite-memory automata}: these are finite automata equipped with a finite number of registers which allow to store symbols during the run, and compare them with the current symbol. The paper \cite{BPT03} extends somehow this idea by allowing transitions which involve an equivalence relation of finite index defined on the set of (vector) values of the registers. The paper \cite{NSV04} continues the study of finite-memory automata, and also introduce pebble automata, which are automata equipped with a finite set of pebbles whose use is restricted by a stack discipline. The automaton can test equality by comparing the pebbled symbols. The work \cite{BDMSS05} addresses decidability issues for some fragment of first-order logic which allows to express properties of words over an infinite alphabet, and introduces a related notion of automaton. More recently, Choffrut and Grigorieff \cite{CG05a} define automata whose transitions are expressed as first-order formulas (see below). Let us also mention the work \cite{DLN05} which studies variants of constraint LTL over infinite domains.

The aim of this paper is to show that a special case of the Feferman-Vaught composition theorem gives rise to a natural notion of automata for finite words over an infinite alphabet, with good closure and decidability properties, as well as several logical characterizations. Building on Mostowski's work \cite{Mos52}, Feferman and Vaught consider in \cite{FV59} several kinds of products of logical structures, and prove that the first-order (shortly: FO) theory of a (generalized) product of structures reduces to the FO theory of the factor structures and the monadic second-order (shortly: MSO) theory of the index structure. We refer the interested reader to the survey papers \cite{Mak04,Tho97} which present several  applications of these results, as well as extensions of the technique; for recent related results see e.g. \cite{Rabinovich04,Rabinovich07,WT04}.

An interesting special case of the Feferman-Vaught (shortly: FV) theorem is when one considers the generalized weak power of a single structure ${\mathfrak M}$, and the index structure is $(\omega;<)$. In this case the domain of the resulting structure roughly consists in the set of finite words over the domain of ${\mathfrak M}$ (seen as an alphabet), and the definable relations can be characterized in terms of automata thanks to B\"uchi-Elgot-Trakhtenbrot results on the equivalence between definability in the MSO theory of $(\omega;<)$ and automata. The automata model and related logics we consider can be seen as direct reformulations of this special case.  Note that the connection between automata and products of structures was already explored in \cite{BlGr00}, where it is shown that automatic structures are closed under finite products.

For the sake of readability, in the paper we first introduce the automata model and prove some of its properties, and then put in evidence the connection with the Feferman-Vaught construction. 

In Section \ref{sectionAutomata} we define the automata model. Given a structure ${\mathfrak M}$ with domain $\Sigma$ (finite or not), we define ${\mathfrak M}$-automata as multitape synchronous finite automata which read finite words over $\Sigma$, and whose transitions are labelled by first-order formulas 
in the language of $\mathfrak M$. We show that the class of relations recognizable by such automata (which are called ${\mathfrak M}$-recognizable relations) are closed under boolean and rational operations, as well as projection, and that the emptiness problem is decidable whenever the FO theory of ${\mathfrak M}$ is. These results are straightforward generalizations of the classical case of a finite alphabet. 

In Section \ref{sectionLogic} we provide two logical characterizations of ${\mathfrak M}$-recognizable languages. The first one uses MSO logic and is an easy adaptation of B\"uchi's classical result \cite{Buchi60}. For the second one, we first introduce the notion of {\em $\mathfrak M$-automatic structures} which extends the notion of automatic structures, and prove some basic related results. Then we extend the Eilenberg-Elgot-Shepherdson FO formalism \cite{EES69} for synchronous relations over words to the case of $\mathfrak M$-recognizable relations. This result, and actually the automaton model itself, are a natural generalization of Choffrut and Grigorieff results mentioned above \cite{CG05a}.

Several results of Section \ref{sectionAutomata}  and \ref{sectionLogic} are rather easy generalizations or reformulations of well-known results; therefore many proofs in these sections are only sketched. 

In Section \ref{sec:FV} we recall useful notions and results about products and powers of structures, then show the close relationship between $\mathfrak M$-recognizability and definability in generalized weak powers. This allows to revisit all previous results in the light of the Feferman-Vaught framework.

Section \ref{sec:applications} presents some applications. We first apply the previous ideas to improve a recent result by Kuske and Lohrey \cite{KL06} related to the monadic chain logic of iteration structures; this application was brought to our attention by Wolfgang Thomas. In the second part of the section, we provide a logical characterization of ${\mathfrak M}$-recognizable relations for the special case where ${\mathfrak M}=(\omega;+)$, in terms of ordinal theories.

In terms of expressive power, ${\mathfrak M}$-automata are incomparable with automata and logics considered in \cite{BDMSS05,KamFra94,NSV04}, since on one hand they allow to express FO constraints, but on the other hand they cannot test whether two positions in a word carry the same symbol (for instance, the language $\{ aa \ | \ a \in \Sigma \}$ is not ${\mathfrak M}$-recognizable whenever $\Sigma$ is infinite, see Example \ref{exa:aa}). As shown e.g. in \cite{BDMSS05,NSV04}  these kinds of tests have to be limited if one wants to keep good decidability properties.  In Section \ref{extFV} we propose a slight extension of the Feferman-Vaught formalism which allows to test whether an $n-$tuple $s_1,\dots,s_n$ of symbols  appearing in distinct positions in a word $w$,  satisfies a formula $\varphi(x_1,\dots,x_n)$ in ${\mathfrak M}$. We isolate a syntactic fragment of this logic, which we denote by $MSO_{R}^+({\mathfrak L})$, for which the satisfiability problem (or in other words, the emptiness problem for related languages) still reduces to the decidability of the FO theory of ${\mathfrak M}$.

\section{Definitions and notations}

In the sequel we deal with finite words over some alphabet, finite or not. Given an alphabet $\Sigma$ (finite or not) we denote by $\Sigma^*$ (respectively $\Sigma^\omega$) the set of finite words (respectively $\omega-$words) over $\Sigma$. The empty word is denoted by $\varepsilon$, and the length of a finite word $w$ by $|w|$. Given a word $w \in \Sigma^*$ with length $n$, we denote by $w[i]$ the $i-$th symbol of $w$ (starting from $i=0$). We shall say that the position $i$ {\em carries} $w[i]$.

We consider several logical formalisms. By FO we mean first-order logic with equality. We shall also consider Monadic Second-Order Logic (shortly: MSO). We denote by $FO({\mathfrak M})$ (respectively $MSO({\mathfrak M})$) the first-order (respectively monadic second-order) theory of the structure ${\mathfrak M}$. We consider only relational structures. Given a language $\mathfrak L$ and a ${\mathfrak L}$-structure ${\mathfrak M}$, for every relational symbol $R$ of $\mathfrak L$ we denote by $R^{\mathfrak M}$ the interpretation of $R$ in ${\mathfrak M}$. However, we will often confuse logical symbols with their interpretation. Moreover we will use freely abbreviations such as $\exists x \in X \ \varphi$.

We shall deal with multitape synchronous automata. As usual, given $n$ finite words $(w_1,\dots,w_n)$ over $\Sigma$, we introduce a padding symbol $\#$, and we complete (if necessary) each $w_i$ with a sufficient number of $\#$'s in order to have words of the same length. Doing this, we obtain $n$ words over $\Sigma \cup \{\#\}$ with the same length, which can be seen as a single word over the alphabet $(\Sigma \cup \{\#\})^n$ (i.e. the alphabet of $n-$tuples of elements of $\Sigma \cup \{\#\}$). This word will be denoted by $\langle w_1,\dots,w_n \rangle$.

Consider a relational language $\mathfrak L$ and  a ${\mathfrak L}$-structure ${\mathfrak M}$ with domain $\Sigma$. Since we have to deal with the symbol $\#$ we shall associate to ${\mathfrak M}$ the structure ${\mathfrak M}_{\#}$ in the extended language ${\mathfrak L}_{\#}={\mathfrak L}\cup\{P_{\#}\}$, such that:

\begin{enumerate}[$\bullet$]
\item the domain of ${\mathfrak M}_{\#}$ is $\Sigma \cup \{\#\}$;
\item for every relational symbol $R$ of $\mathfrak L$, we have $R^{{\mathfrak M}_{\#}}=R^{\mathfrak M}$; 
\item $P_{\#}(x)$ holds in ${\mathfrak M}_{\#}$ if and only if $x=\#$.  
\end{enumerate}

\section{Definition and properties of ${\mathfrak  M}$-automata}\label{sectionAutomata}

In this section we introduce the notion of ${\mathfrak  M}$-automata and ${\mathfrak  M}$-recognizable relations, and prove some basic results.

Let $\Sigma$ denote an alphabet, finite or not, and let ${\mathfrak  M}$ denote an $\mathfrak L-$structure with domain $\Sigma$. An {\em ${\mathfrak M}$-automaton} is a finite $n-$tape synchronous non-deterministic automaton which reads finite words over $\Sigma$. Transition rules are triplets of the form $(q,\varphi,q')$, where $q,q'$ are states of the automaton, and $\varphi(x_1,\dots,x_n)$ is a first-order formula in the language ${\mathfrak L}_{\#}$ of ${\mathfrak M}_{\#}$. The transition $(q,\varphi,q')$ can be executed if the $n-$tuple of current symbols read by the $n$ heads satisfies $\varphi$ in ${\mathfrak M}_\#$.

\begin{defi}
Let $\Sigma$ be an alphabet and let ${\mathfrak  M}$ denote an $\mathfrak L-$structure with domain $\Sigma$. An ${\mathfrak M}$-automaton is defined as a $7-$tuple $\mathcal{A}=(Q,n,\Sigma,{\mathfrak M},E,I,T)$ where
\begin{enumerate}[$\bullet$]
\item $Q$ is a finite set (of states);
\item $n \geq 1$ is the number of tapes;
%\item $\Sigma$ is an alphabet (finite or not);
\item $E \subseteq Q \times {\mathcal F}_n \times Q$ is the set of transitions, where ${\mathcal F}_n$ denotes the set of ${\mathfrak L}_{\#}$-formulas with $n$ free variables;
\item $I \subseteq Q$ is the set of initial states;
\item $T \subseteq Q$ is the set of terminal states.
\end{enumerate}

Given an $n-$tuple $w = (w_1,\dots ,w_n)$ of words over $\Sigma$, a {\em path $\gamma$ in $\mathcal A$ labeled by $\langle w \rangle$} is a sequence of states $\gamma = (q_0, \dots ,q_{m})$, such that $m=|\langle w 
\rangle|$, $q_0 \in I$, and for every $i <m$ there exists a ${\mathfrak L}_{\#}$-formula $\varphi(x_1,\dots,x_n)$ such that $(q_i,\varphi,q_{i+1})\in E$ and 
$${\mathfrak M}_\# \models \varphi(\pi_1(\langle w \rangle)[i],\dots,\pi_n(\langle w\rangle)[i])$$ 
where $\pi_j(\langle w\rangle)$ denotes the $j-$th component of $\langle w\rangle$. The path $\gamma$ is {\em successful} if $q_m \in T$. We say that $w$ is accepted by $\mathcal{A}$ if $\langle w\rangle$ is the label of some successful path. We denote by $L(\mathcal A)$ the set of words $w \in (\Sigma^*)^n$ which are accepted by $\mathcal A$.  
 
\end{defi}

\begin{defi}
Let $n \geq 1$. A relation $X \subseteq (\Sigma^*)^n$ is said to be {\em ${\mathfrak M}$-recognizable} if and only if there exists an $\mathfrak M$-automaton $\mathcal A$ with $n$ tapes such that $X=L(\mathcal A)$.   
\end{defi}

\begin{exa} 
 Let ${\mathfrak M}=(\omega;+)$ where $+$ denotes the graph of addition. The following relations are ${\mathfrak M}$-recognizable:
\begin{enumerate}[(1)]
\item the set of words over $\omega$ (seen as as an infinite alphabet) of the form $(1,0,...,0)$ (we allow the case where there is no $0$). Consider indeed the $\mathfrak M$-automaton with two states $q_0,q_1,$ where $q_0$ is initial and $q_1$ is terminal, and whose set of transitions is $\{(q_0,\varphi_1,q_1),(q_1,\varphi_0,q_1)\}$ where $\varphi_0(x)$ is the formula $x+x=x$, and $\varphi_1(x)$ expresses that $x=1$. The automaton is pictured in Figure
  \ref{fig:basicex}. 
  
\begin{figure}[htbp]
  \begin{center}
    \gasset{Nw=6,Nh=6,loopdiam=7}
    \begin{picture}(35,20)(0,-5)
    \node[Nmarks=i](1)(10,0){$q_0$}
    \node[Nmarks=f](2)(25,0){$q_1$}
    \drawedge(1,2){$\varphi_1$}
    \drawloop(2){$\varphi_0$}
    \end{picture}
    \caption{A simple ${\mathfrak M}$-automaton}
    \label{fig:basicex}
  \end{center}
\end{figure}
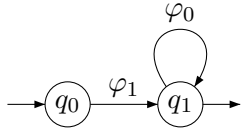

\item the set of words over $\omega$ whose symbols are alternatively even and odd. Consider indeed the $\mathfrak M$-automaton with two states $q_0,q_1,$ where $q_0,q_1$ both are initial and terminal states, and whose set of transitions is $\{(q_0,\varphi_e,q_1),(q_1,\varphi_o,q_0)\}$ where $\varphi_e(x)$ is the formula $\exists z \ z+z=x$, and $\varphi_o(x)= \neg \varphi_e(x)$.

\item the relation $L\subseteq \omega^* \times \omega^* \times \omega^*$ defined by $(u,v,w) \in L$ if and only if $u,v,w$ have the same length, and moreover for every $i$ the $i-$th symbol of $w$ equals the sum of the corresponding symbols of $u$ and $v$. Consider indeed the $\mathfrak M$-automaton with a single state $q$ (which is initial and terminal) and a single transition $(q,\varphi,q)$ where $\varphi(x,y,z)$ is the formula $x+y=z$.
Observe that if $u,v$ and $w$ do not have the same length then the last letter of $\langle u,v,w \rangle $, say $(u_m,v_m,w_m)$, has at least one component which is equal to $\#$, which by the very definition of ${\mathfrak M}_\#$  
 implies ${\mathfrak M}_\# \not\models \varphi(u_m,v_m,w_m)$, which implies in turn that there is no (successful) run of $\mathcal A$ labelled by $\langle u,v,w \rangle$.

\end{enumerate}
\end{exa}

\begin{exa}

Let ${\mathfrak M}=(\Sigma; (P_a)_{a \in \Sigma})$, where $P_a(x)$ holds if and only if $x=a$.  One can show that if $\Sigma$ is finite then ${\mathfrak M}$-recognizable relations coincide with synchronous relations (as defined in \cite{EES69}).

\end{exa}

\begin{exa}\label{exa:aa} For every ${\mathfrak L}$-structure ${\mathfrak M}=(\Sigma;...)$ such that $\Sigma$ is infinite, the language $X=\{aa \ | \ a\in \Sigma \}$ is not ${\mathfrak M}$-recognizable. Indeed assume for a contradiction that there exists some ${\mathfrak M}$-automaton $\mathcal{A}=(Q,n,\Sigma,{\mathfrak M},E,I,T)$ which accepts $X$. Since $X$ is infinite and $E$ is finite, there exists an infinite subset of $X$ whose elements admit a common successful path. More precisely, there exist two transitions $(q_0,\varphi_1,q_1),(q_1,\varphi_2,q_2) \in E$ such that $q_0 \in I$, $q_2 \in T$, and  infinitely many elements $a \in \Sigma$ satisfy 
${\mathfrak M}_\# \models \varphi_1(a)$ and  ${\mathfrak M}_\# \models \varphi_2(a)$. Thus there exist at least two distinct elements $a_1 \ne a_2$ such that ${\mathfrak M}_\# \models \varphi_1(a_1)$ and  ${\mathfrak M}_\# \models \varphi_2(a_2)$, which implies that $a_1a_2$ is accepted by $\mathcal A$, and this leads to a contradiction.
\end{exa}

The closure properties of synchronous relations still hold for $\mathfrak M$-recognizable relations.

\begin{prop}\label{prop:cloture}
The class of ${\mathfrak M}$-recognizable relations is closed under 
\begin{enumerate}[\em(1)]
\item boolean operations;
\item cylindrification;
\item projection.
\end{enumerate}
\end{prop}

\begin{proof}{(sketch)} $(1)$ the closure under union is a straightforward adaptation of the classical construction for non-deterministic automata. Let us outline the proof for the closure under complementation. Assume that the relation $R \subseteq (\Sigma^*)^n$ is recognized by the $\mathfrak M$-automaton $\mathcal{A}$. Let $\varphi_1,\dots,\varphi_m$ denote the formulas which appear in the transitions of $\mathcal A$. Consider, for every subset $J \subseteq \{1,\dots,m\}$, the formula $\psi_J: \ \bigwedge_{i \in J} \varphi_i \wedge \bigwedge_{i \not\in J} \neg\varphi_i.$  
The $\mathfrak M$-automaton $\mathcal{A}'$ which is defined from $\mathcal{A}$ by replacing every transition $(q,\varphi_i,q') \in E$ with all transitions of the form  $(q,\psi_J,q')$ where $i \in J$, also recognizes $R$. Moreover $\mathcal{A}'$ can be seen as a classical non-deterministic automaton over the finite ``alphabet" of formulas $\psi_J$, and thus it can be determinized, i.e. transformed into an equivalent $\mathfrak M$-automaton $\mathcal{A}''$ whose transitions involve the formulas $\psi_J$, and such that for every state $q$ and every formula $\psi_J$ there exists a single transition of $\mathcal{A}''$ the form $(q,\psi_J,q')$. Now one can use the usual construction for complementation of deterministic automata, i.e. turn non-terminal states to terminal states and conversely, and get an $\mathfrak M$-automaton which recognizes $(\Sigma^*)^n \setminus R$.

$(2)$ is straightforward.

For $(3)$, in order to recognize the projection of $R \subseteq (\Sigma^*)^n$, say over the $n-1$ first components, it suffices to replace, in the $\mathfrak M$-automaton which recognizes $R$, all transitions ~$(q,\varphi(x_1,\dots,x_n),q')$ with transitions $(q,\exists x_n \varphi(x_1,\dots,x_n),q')$. 
\end{proof}

Regarding the emptiness problem for ${\mathfrak M}$-recognizable languages, the main difference with the classical case is that in an ${\mathfrak M}$-automaton $\mathcal A$ there can exist transitions $(q,\varphi,q')\in E$ such that no $n-$tuple of elements of ${\mathfrak M}_{\#}$ satisfies $\varphi$; such transitions will never be executed by $\mathcal A$.  Thus one has to remove such transitions from $E$ in order to apply the usual reachability algorithm for the emptiness problem; this can be done effectively if and only if $FO({\mathfrak M}_{\#})$ is decidable. Since $FO({\mathfrak M}_{\#})$ and $FO({\mathfrak M})$ are reducible to each other, we get the following result.

\begin{prop}\label{prop:reducdecid}
The decidability of the emptiness problem for ${\mathfrak M}$-recognizable languages is equivalent to the decidability of $FO({\mathfrak M})$. 
\end{prop}

\section{Logic and ${\mathfrak M}$-automata}\label{sectionLogic}

There exist three important logical formalisms which capture automata: 

\begin{enumerate}[$\bullet$]
\item B\"uchi-Elgot-Trakhtenbrot MSO logic (see \cite{Buchi60}), i.e. the weak monadic second order theory of $(\omega,<)$;

\item the Eilenberg-Elgot-Shepherdson (shortly: EES) formalism \cite{EES69}, i.e. the FO theory of  
${\mathcal S}=(\Sigma^*;EqLength,\preceq, \{L_a\}_{a\in \Sigma})$
where
\begin{enumerate}[-]
\item $EqLength(x,y)$ holds if and only if $x$ and $y$ have the same length
\item $x\preceq y$ holds if and only if $x$ is a prefix of $y$
\item $L_a(x)$ holds if and only if $a$ is the last letter of $x$. 
\end{enumerate}

\item the so-called B\"uchi Arithmetic of base $k$, i.e. the FO theory of the structure $(\omega;+,V_k)$ where $V_k(x)$ denotes the greatest power of $k$ which divides $x$, see \cite{BHMV94}.
\end{enumerate}

In this section we extend the two first formalisms to the case of words over any alphabet (finite or not). In order to extend the EES formalism, we introduce the notion of $\mathfrak M$-automatic structure, which generalizes the one of automatic structures.

 In Section \ref{subsection:ordinal} we will prove that for ${\mathfrak M}=(\omega;+)$, the class of ${\mathfrak M}$-recognizable relations corresponds to the class of relations definable in the structure $(\omega^\omega; +)$. The latter structure can therefore be seen as ``B\"uchi Arithmetic of base $\omega$".

\subsection{Monadic Second-Order Logic}

\

B\"uchi, Elgot and Trakhtenbrot prove that languages of words definable by MSO logic coincide with regular languages (see \cite{Tho97handbook}). As an example if $\Sigma=\{a,b\}$ then one can characterize the set of words $w \in aa^*b^*$ in MSO logic with the formula
$$\exists x (Q_a(x) \wedge \forall y ((x<y \rightarrow Q_b(y)) \wedge (y<x  \rightarrow Q_a(y))))$$
where the (first-order) variables $x,y$ are interpreted as positions in the word, $<$ denotes the natural ordering of positions, and $Q_s(y)$ holds if and only if the $y-$th position in the word carries the symbol $s$. The unary predicates $Q_a$ and $Q_b$ express properties related to elements of $\Sigma$. These properties are actually first-order definable in the structure ${\mathfrak M}=(\Sigma; P_a, P_b)$. We shall extend this formalism by considering any structure ${\mathfrak M}$ with domain $\Sigma$ (finite or not) and adding to the MSO formalism unary predicates $\alpha_F(x)$ which express that the symbol at position $x$ satisfies the formula $F$ in ${\mathfrak M}$.  

More formally, let ${\mathfrak M}=(\Sigma;...)$ be an ${\mathfrak L}$-structure. We associate to every ${\mathfrak L}_{\#}$-formula $F$ some unary relational symbol $\alpha_F$. We define then $MSO({\mathfrak L})$ as MSO over the language $\{<,(\alpha_F)_{F \in {\mathcal F}}\}$ where $\mathcal F$ denotes the set of (first-order) ${\mathfrak L}_{\#}$-formulas with at least one free variable.

\begin{defi}
We say that $A \subseteq (\Sigma^*)^n$ is {\em $MSO({\mathfrak M})$-definable} if there exists an $MSO({\mathfrak L})$-sentence $\psi$ such that $w=(w_1,\dots,w_n) \in A$ if and only if $$(D,<_D,(\alpha_F)_{F \in {\mathcal F}})\models \psi$$ where 
\begin{enumerate}[$\bullet$]
\item $D=\{0,1,\dots,|\langle w \rangle|-1\}$, and $<_D$ is the natural ordering relation restricted to $D$;
\item For every ${\mathfrak L}_{\#}$-formula $F$ with $n$ free variables, and every position $x$ in $w$, the formula $\alpha_F(x)$ holds in $(D,<_D,(\alpha_F)_{F \in {\mathcal F}})$  if and only if 
$${\mathfrak M}_\# \models F(\pi_1(\langle w \rangle)[x],\dots,\pi_n(\langle w \rangle)[x]),$$
where $\pi_i(\langle w \rangle)$ denotes the $i-$th component of $\langle w\rangle$.
\end{enumerate}
\end{defi}

\begin{exa}\label{exa:msodef}
Let ${\mathfrak M}=(\omega;+)$.
\begin{enumerate}[$\bullet$]
\item the set $A \subseteq \omega^*$ of words that contain only even symbols is $MSO({\mathfrak M})$-definable by the sentence $\forall y \  \alpha_F(y)$,  
where $F(x): \ \exists z (z+z = x)$
\item the set $B \subseteq \omega^*$ of words whose symbols are alternatively even and odd is $MSO({\mathfrak M})$-definable by the formula 
$$\forall x \forall y ((x<y \wedge \neg \exists z \ x<z<y) \rightarrow ((\alpha_F(x) \leftrightarrow  \alpha_{\neg F}(y)) ))$$
where $F$ is the formula defined above.
\end{enumerate}
\end{exa}

\begin{exa}
Let ${\mathfrak M}=(\Sigma; (P_a)_{a \in \Sigma})$. If $\Sigma$ is finite, then $MSO({\mathfrak M})$-definable languages coincide with languages definable in B\"uchi's MSO logic. Indeed in this case it is rather easy to prove that any formula $F(x_1,\dots,x_n)$ in the language of ${\mathfrak M}$ is equivalent to a boolean combination of formulas of the form $P_a(x_i)$, which implies in turn that every predicate $\alpha_F(y)$ is equivalent to a boolean combination of predicates $Q_a(y)$.

\end{exa}

We now generalize the equivalence between recognizability and definability to the case of any alphabet $\Sigma$ (finite or not). 

\begin{prop}\label{logicauto}

Let ${\mathfrak M}$ be a   structure with domain $\Sigma$. For every $n\geq 1$ and every relation $R \subseteq (\Sigma^*)^n$, the relation $R$ is $MSO({\mathfrak M})$-definable if and only if it is ${\mathfrak M}$-recognizable.

\end{prop}

\begin{proof}
The proof is a simple adaptation from B\"uchi's equivalence between WMSO definability and recognizability \cite{Buchi60}. For the direction from recognizability
to definability one uses the B\"uchi's technique of encoding of an accepting run of an automaton by a formula. For the converse one uses the fact that the formulas $\alpha_{F_1},\dots,\alpha_{F_m}$ appearing in an $MSO({\mathfrak L})$-sentence $\psi$  can be chosen such that 
every $n-$tuple of elements of the domain of ${\mathfrak M}_{\#}$ satisfies exactly one formula among the $F_i$'s in ${\mathfrak M}_{\#}$, which allows then to see words over $\Sigma \cup \{\#\}$ as words over the finite alphabet $\{1,\dots,m\}$ and then prove the result by induction on the construction of $\psi$. 
\end{proof}

\

\subsection{${\mathfrak  M}$-automatic structures and an extension of the EES formalism}\label{automatic}

Automatic structures (see \cite{Hod83,KhoNer94,BlGr00}) are relational structures which can be presented by finite automata over a finite alphabet. 

\begin{defi}
The structure ${\mathfrak  N}=(N;R_1,\dots,R_k)$ is said to be {\em automatic} if there exist a finite alphabet $\Sigma$ and an injective mapping $\mu:N \to \Sigma^*$ such that the images by $\mu$
of $N,R_1,\dots,R_k$ are synchronous relations.
\end{defi}

The fundamental result about automatic structures is the following.
\begin{thm}[\cite{Hod83}]\label{thm:hodgson}
If ${\mathfrak  N}$ is automatic then:
\begin{enumerate}[\em(1)]
\item the image by $\mu$ of every relation definable in ${\mathfrak  N}$ is a synchronous relation;
\item $FO({\mathfrak  N})$ is decidable.
\end{enumerate}
\end{thm}

We can generalize the previous notions and results to the case of an infinite alphabet.

\begin{defi}
Let ${\mathfrak M}=(\Sigma;\dots)$ and ${\mathfrak  N}=(N;R_1,\dots,R_k)$ be two   structures. We say that ${\mathfrak  N}$ is {\em ${\mathfrak  M}$-automatic} if there exists an injective mapping $\mu:N \to \Sigma^*$ such that the images by $\mu$
of $N,R_1,\dots,R_k$ are ${\mathfrak M}$-recognizable relations.
\end{defi}

\begin{exa}\label{exa:skolemauto}
Let ${\mathfrak M}=(\omega;+)$, and let ${\mathfrak  N}=(\omega \setminus\{0\}; \times)$ where $\times$ denotes the graph of multiplication. The structure ${\mathfrak  N}$ (which is often called {\em Skolem arithmetic}) is ${\mathfrak  M}$-automatic. Consider indeed the function $\mu: (\omega \setminus\{0\}) \to \omega^*$ which maps every natural number $n>1$ whose prime decomposition is $n= p_0^{n_0}p_1^{n_1}\dots p_k^{n_k}$, where $p_i$ denotes the $i-$th prime and $n_k \ne 0$, to the word $\mu(n)=n_0 n_1\dots n_k$. Moreover let $\mu(1)=\varepsilon$. It is not difficult to check that the image by $\mu$ of $\omega \setminus\{0\}$ and $\times$ are $\mathfrak M$-recognizable.
\end{exa}

One can prove rather easily that if $\Sigma$ is finite and ${\mathfrak M}=(\Sigma;(P_a)_{a \in \Sigma})$, then ${\mathfrak  M}$-automatic structures correspond to automatic structures. On the other hand, there exist structures which are $\mathfrak M$-automatic but not automatic. For instance in \cite{BlGr00} it is proven that the structure ${\mathfrak  N}$ considered in the previous example, i.e. Skolem arithmetic, is not automatic\footnote{Note that Skolem arithmetic is tree-automatic in the sense of \cite{BlGr00}. This comes from the fact that the structure $\mathfrak M$, i.e. Presburger arithmetic, is automatic. More generally one can prove that if ${\mathfrak  M}$ is an automatic structure and  ${\mathfrak  N}$ is ${\mathfrak  M}$-automatic then ${\mathfrak  N}$ is tree-automatic.}.

We can extend Theorem \ref{thm:hodgson} in the following way.

\begin{thm}\label{thm:hodgson2}
If ${\mathfrak  N}$ is ${\mathfrak  M}$-automatic then:
\begin{enumerate}[\em(1)]
\item the image by $\mu$ of every relation definable in ${\mathfrak  N}$ is ${\mathfrak M}$-recognizable;
\item $FO({\mathfrak  N})$ reduces to $FO({\mathfrak  M})$.
\end{enumerate}
\end{thm}

\begin{proof}(sketch)
This is a straightforward adaptation of the proof of Theorem \ref{thm:hodgson}. 
For $(1)$ one proceeds by induction on the formulas and use the closure properties of $\mathfrak M$-recognizable relations as stated in Proposition \ref{prop:cloture}. 
For $(2)$, one first deduces from $(1)$ that the decidability of $FO({\mathfrak  N})$ reduces to the decidability of the emptiness problem for $\mathfrak M$-automata, and then uses Proposition \ref{prop:reducdecid}.
\end{proof}

Let us now turn to the first-order characterization of $\mathfrak M$-recognizability. Eilenberg, Elgot and Shepherdson prove the following result.

\begin{thm}[\cite{EES69}] Let $\Sigma$ denote a finite alphabet with at least two elements. For every $n \geq 1$, an $n-$ary relation over $\Sigma$ is synchronous if and only if it is definable in the structure 
$${\mathcal S}=(\Sigma^*;EqLength,\preceq, \{L_a\}_{a\in \Sigma})$$
where
\begin{enumerate}[$\bullet$]
\item $EqLength(x,y)$ holds if and only if $x$ and $y$ have the same length;
\item $x\preceq y$ holds if and only if $x$ is a prefix of $y$;
\item $L_a(x)$ holds if and only if $a$ is the last letter of $x$. 
\end{enumerate}
\end{thm}

It is easy to check that for every finite alphabet $\Sigma$ the structure $\mathcal S$ is automatic, and moreover one can deduce from the above theorem that every automatic structure is indeed $FO-$interpretable in $\mathcal S$ (see e.g. \cite{BlGr00}).

In \cite{EES69}, the authors asked whether there is an appropriate notion of automata that captures this logic when $\Sigma$ is infinite. Choffrut and Grigorieff \cite{CG05a} recently solved this problem (and also other questions raised in \cite{EES69}) by introducing a notion of automata with constraints expressed as FO formulas. It appears that the automata notion they consider captures exactly ${\mathfrak M}$-recognizable relations for  the special case ${\mathfrak M}=(\Sigma;(P_a)_{a \in \Sigma})$.  

We can generalize the previous results in the following way.

\begin{defi}\label{defEES}
Let ${\mathfrak M}=(\Sigma;R_1,\dots,R_k)$ be a structure. We define the stucture
$$S_{\mathfrak M}=(\Sigma^*;EqLength,\preceq,A_{R_1},\dots,A_{R_k},A_=)$$
where
\begin{enumerate}[$\bullet$]
\item $EqLength(x,y)$ holds if and only if $x$ and $y$ have the same length;
\item $x\preceq y$ holds if and only if $x$ is a prefix of $y$;
\item for every $i$, $A_{R_i}(x_1,\dots,x_n)$ holds if and only if
there exist words $w_1,\dots,w_n \in \Sigma^*$ and symbols $a_1,\dots,a_n \in \Sigma$ such that: 
\begin{enumerate}[-]
\item $x_i=w_ia_i$ for every $i$;
\item all $w_i$'s have the same length;
\item $(a_1,\dots,a_n) \in R_i^{\mathfrak M}$.
\end{enumerate}
\item $A_{=}(x,y)$ holds if and only if $x$ and $y$ have the same length and the same last letter.
\end{enumerate}
\end{defi}

\begin{thm}\label{caractEES} Let ${\mathfrak M}=(\Sigma;R_1,\dots,R_k)$ be a structure, where $|\Sigma|\geq 2$.
\begin{enumerate}[\em(1)]
\item For every $n \geq 1$ and every $n-$ary relation $R$ over $\Sigma$, the relation $R$ is ${\mathfrak M}$-recognizable if and only if it is definable in $S_{\mathfrak M}$;
\item $FO(S_{\mathfrak M})$ reduces to $FO({\mathfrak M})$.
\end{enumerate}
\end{thm}

\begin{proof}
\hfill
\begin{enumerate}[(1)]
\item It is easy to check that all base relations of $S_{\mathfrak M}$ are ${\mathfrak M}$-recognizable, which implies that $S_{\mathfrak M}$ is $\mathfrak M$-automatic, and thus by Theorem \ref{thm:hodgson2} every relation definable in $S_{\mathfrak M}$ is ${\mathfrak M}$-recognizable.

For the converse one can adapt again B\"uchi's technique of encoding runs of automata by monadic second-order variables.
Assume that $R \subseteq (\Sigma^*)^n$ is recognized by some $\mathfrak M$-automaton $\mathcal A=(Q,n,\Sigma,{\mathfrak M},E,I,T)$ whose set of states is $Q=\{q_0,\dots,q_k\}$. We can define $R$ in $S_{\mathfrak M}$ by a formula $\varphi(w_1,\dots,w_n)$ which expresses the existence of a successful path of $\mathcal A$, say $(q_{j_0},q_{j_1},\dots,q_{j_{m}})$, labelled by $\langle w_1,\dots,w_n \rangle$. The formula encodes the path with $k+2$ words $z_0,z_1,\dots,z_k,y$ whose length is $m+1$ (that is, $|\langle w_1,\dots,w_n \rangle|+1$), and such that for every $i$, the word $z_{j_i}$ is the only word among $z_0,z_1,\dots,z_k$ whose $i-$th symbol equals the $i-$th symbol of $y$. 
That is, the variable $y$ serves to identify $z_{j_i}$ for every position $i$ (with the help of the predicate $A_{=}$).

\item This is an immediate consequence of Theorem \ref{thm:hodgson2} together with the fact that $S_{\mathfrak M}$ is $\mathfrak M$-automatic.
\end{enumerate}

\end{proof}

In Section \ref{sec:applications} we will improve item $(2)$ of the above theorem by proving that even the monadic second-order chain logic of $S_{\mathfrak M}$ reduces to $FO({\mathfrak M})$.

\section{A special case of the Feferman-Vaught composition theorem}\label{sec:FV}

In this section we put in evidence the strong relationship between ${\mathfrak M}$-automata and a special case of  the Feferman-Vaught composition theorem. The Feferman-Vaught method presented in \cite{FV59} generalizes Mostowski's work \cite{Mos52} about products of structures. Let us recall some useful notions and results from \cite{Mos52,FV59}.

%We now recall the Feferman-Vaught composition method \cite{FV59}. More precisely we shall focus on the notion of generalized power introduced in \cite{FV59}, which generalizes the notion of power of structures studied by Mostowski \cite{Mos52}. 

%Mostowski defines 
%$\prod_I {\mathfrak M}$
\begin{defi}\label{powmost}
Let ${\mathfrak M}=(\Sigma;R^{\mathfrak M}_1,\dots,R^{\mathfrak M}_k)$ be a   structure, and let $I$ be a non empty set.  
The {\em direct power} of a ${\mathfrak M}$ with respect to $I$ is defined as the structure 
$${\mathfrak N}=(\Sigma^I; R^{\mathfrak N}_1,\dots,R^{\mathfrak N}_k)$$ such that 
\begin{enumerate}[$\bullet$]
\item the domain of $\mathfrak N$ is the set $\Sigma^I$ of sequences $f:I \to \Sigma$
\item for every $j \in \{1,\dots,k\}$, if $R_j$ is $n-$ary then for every $n-$tuple $(f_1,\dots,f_n)$ of elements of $\Sigma^I$, we have $(f_1,\dots,f_n) \in R^{\mathfrak N}_j$ if and only if 
$(f_1(i),\dots,f_n(i)) \in R^{\mathfrak M}_j$ for every $i \in I$.
\end{enumerate}
\end{defi}

\begin{exa}\label{ex:omega+}
Let ${\mathfrak M}=(\omega;+)$, and $I=\omega$. The direct power of $\mathfrak M$ with respect to $I$ is the structure $\mathfrak N$ with domain the set of sequences $f:\omega \to \omega$ (in other words, the set of $\omega-$words over the alphabet $\omega$), and such that $(f_1,f_2,f_3) \in +^{\mathfrak N}$ if and only if $f_1(i)+f_2(i)=f_3(i)$ for every $i \in \omega$. 
\end{exa}

Mostowski proves that the evaluation of FO formulas in the direct power ${\mathfrak N}$ reduces to the evaluation of formulas in the factor structure $\mathfrak M$ and formulas in the structure ${\mathfrak S}=(S(I);\subseteq)$ (the {\em index structure}) where $S(I)$ denotes the power set of $I$, and $\subseteq$ is interpreted as the inclusion relation. Note that the FO theory of $\mathfrak S$ is a variant of the MSO theory of $I$.

\begin{defi}\label{def:reduction}
Let $R$ be an $m-$ary relation over elements of $\Sigma^I$. A {\em reduction sequence} for $R$ (with respect to the structures $\mathfrak{M,S}$) is a sequence $\xi=(G,\theta_1,\dots,\theta_l)$ such that
\begin{enumerate}[$\bullet$]
\item $G$ is a formula in the language of $\mathfrak S$;
\item $\theta_1,\dots,\theta_l$ are formulas in the language of $\mathfrak M$;
\item for every $m$-tuple $(f_1,\dots,f_m)$ of elements of $\Sigma^I$, 
we have $(f_1,\dots,f_m) \in R$ 
if and only if 
$${\mathfrak S} \models G(T_1,\dots,T_l)$$ 
where 
$$T_i=\big\{ x \in I \ | \  {\mathfrak M} \models \theta_i(f_1(x),\dots,f_m(x)) \} {\textnormal{ \ \ for 
\ every \ }} i \in \{1,\dots,l 
\big\}.$$ 
\end{enumerate}
\end{defi}\medskip

\noindent{\bf Example \ref{ex:omega+} (continued).} 
\begin{enumerate}[$\bullet$]
\item The base relation $+^{\mathfrak N}$ of $\mathfrak N$ admits a reduction sequence with respect to ${\mathfrak M} = (\omega;+)$ and ${\mathfrak S}=(S(\omega);\subseteq)$. Indeed we have $(f_1,f_2,f_3) \in +^{\mathfrak N}$ if and only if the set of indexes $i$ such that $f_1(i)+f_2(i)=f_3(i)$ equals $\omega$, that is if 
$${\mathfrak S} \models \forall Y \  Y \subseteq T $$ 
where 
$$T=\{ i \in \omega \ | \  {\mathfrak M} \models f_1(i)+f_2(i)=f_3(i) \}.$$
Thus $+^{\mathfrak N}$ admits the reduction sequence $\xi=(G,\theta_1)$, where ~$G(X): \forall Y \ Y \subseteq X$ and $\theta_1(x,y,z): x+y=z$.

More generally all base relations of a direct power of a structure $\mathfrak M$ admit a reduction sequence. 
\item Consider the formula 
$$F(x): \forall y_1 \forall y_2((\exists z_1 \ y_1+z_1=x \wedge \exists z_2 \ y_2+z_2=x)\rightarrow \exists u (y_1+u=y_2 \vee y_2+u=y_1))$$
We have ${\mathfrak N} \models F(f)$ if and only if $f$ admits at most one non-null element $f(i)$, that is, if the set of elements $i$ such that $f(i)+f(i) \ne f(i)$ contains at most one element. The unary relation defined by $F$ admits the reduction sequence $\xi=(G,\theta_1)$ where 
$$G(x):  \forall Y \forall Z ((Y \subseteq X \wedge Z \subseteq X) \rightarrow  (Y \subseteq Z \vee Z \subseteq Y))$$
and $\theta_1(x): x+x \ne x$.  
\end{enumerate}

\begin{thm}[Mostowski \cite{Mos52}]\label{MosTh}
Let $\mathfrak M$ be a   structure, $I$ be a non empty set, and let $\mathfrak N$ be the direct power of $\mathfrak M$ with respect to $I$. Then
\begin{enumerate}[\em(1)]

\item one can compute effectively a reduction sequence for every relation which is FO definable in $\mathfrak N$; 
\item the FO theory of ${\mathfrak  N}$ reduces to the FO theories of ${\mathfrak  M}$ and ${\mathfrak  S}$.
\end{enumerate}
\end{thm}

Mostowski also proves that $FO({\mathfrak  S})$ is decidable for every set $I$ (by elimination of quantifiers), which together with point $(2)$ in the above theorem implies that for every $I$ the FO theory of the power of ${\mathfrak M}$ with respect to $I$ reduces to the FO theory of ${\mathfrak M}$.\bigskip

\noindent{\bf Example \ref{ex:omega+} (continued).}  The FO theory of
$\mathfrak N$ is decidable, since it reduces to the FO theory of
$(\omega;+)$ which is decidable \cite{Pre29}.\medskip

Another important notion is the one of {\em weak direct power} of a structure. In this variant we consider a   structure ${\mathfrak M}=(\Sigma;R^{\mathfrak M}_1,\dots,R^{\mathfrak M}_k,P_e^{\mathfrak M})$ with some distinguished element $e \in \Sigma$, and where $P_e(x)$ holds in $\mathfrak M$ if and only if $x=e$. The weak power of ${\mathfrak M}$ with respect to $I$ is defined in the same way as in Definition \ref{powmost} but here the domain of ${\mathfrak N}$ is the set ${\mathfrak M}^{(I)}_e$ of sequences $f: I \to \Sigma$ such that $f(i)\ne e$ for finitely many values of $i$. Mostowski proves that Theorem \ref{MosTh} still holds for weak direct powers, with the following modifications:
\begin{enumerate}[$\bullet$]
\item for the index structure $\mathfrak S$ one considers the structure ${\mathfrak S}_{fin}=(S^+(I);\subseteq)$ where $S^+(I)$ denotes the set of finite subsets of $I$. 
\item one considers only reduction sequences $\xi=(G,\theta_1,\dots,\theta_l)$ such that ${\mathfrak M} \models \neg \theta_i(e,\dots,e)$ for every $i$ (this condition ensures that all sets $T_i$ in Definition \ref{def:reduction} are finite).
\end{enumerate}

Note that 
${\mathfrak S}_{fin}$ is a FO variant of the weak MSO theory of $I$. In \cite{Mos52} it is shown that $FO({\mathfrak S}_{fin})$ is decidable for every set $I$. Therefore $FO(\mathfrak N)$ reduces to 
$FO(\mathfrak M)$.

\begin{exa}{(Decidability of Skolem arithmetic \cite{Mos52})}\label{exa:skolem} We revisit here Example \ref{exa:skolemauto}. 
Consider the structure ${\mathfrak M}=(\omega;+,P_0)$ where $0$ is the distinguished element, and $I=\omega$. Then the weak direct power of $\mathfrak M$ with respect to $I$ is the structure ${\mathfrak N}= (\omega^{(\omega)}_0;+,P_0)$ whose domain is the set of sequences $f:\omega \to \omega$ such that $f(i)\ne 0$ for finitely many values of $i$, $+$ denotes the graph of addition of sequences $f:\omega \to \omega$, and $P_0(f)$ holds only for $f=0$. It follows from Mostowski's result that $FO({\mathfrak N})$ is decidable since it reduces to $FO({\mathfrak M})$ which is decidable \cite{Pre29}. Now observe that the application $h: \omega^{(\omega)}_0 \to \omega\setminus \{0\}$ which maps every sequence $f \in \omega^{(\omega)}_0$ to the integer $h(f)=2^{f(0)}3^{f(1)}\dots,$ defines an isomorphism between ${\mathfrak N}$ and the structure $(\omega\setminus \{0\};\times,P_1)$ where $\times$ denotes the graph of multiplication, and $P_1(x)$ holds if and only if $x=1$. Therefore the FO theory of the latter structure  is decidable.
\end{exa}

Feferman and Vaught \cite{FV59} generalize Mostowski's technique by allowing index structures of the form ${\mathfrak S}=(S(I);\subseteq,S_1,\dots,S_m)$ where the $S_i$'s denote any relations.

\begin{defi}
Let ${\mathfrak M}=(\Sigma;R^{\mathfrak M}_1,\dots,R^{\mathfrak M}_k)$ be a   structure, $I$ be a set, and let ${\mathfrak S}=(S(I);\subseteq,S_1,\dots,S_m)$ where the $S_i$'s denote relations. We call {\em generalized power\footnote{Our definition is a slight modification of the original definition. Indeed Feferman and Vaught define {\em the} generalized power of $\mathfrak M$ with respect to $\mathfrak S$ as the structure with domain $\Sigma^I$ and with infinitely many relations $P_i$, one for each relation which admits a reduction sequence.}
 of ${\mathfrak M}$ with respect to ${\mathfrak S}$} every   structure of the form  
$${\mathfrak N}=(\Sigma^I; P_{1}, P_2, \dots, P_n )$$
such that all relations $P_i$ admit a reduction sequence with respect to $\mathfrak M$ and $\mathfrak S$.
\end{defi}

\begin{thm}[Feferman-Vaught \cite{FV59}]\label{fvthm}
If $\mathfrak N$ is a generalized power of $\mathfrak M$ with respect to $\mathfrak S$, then  
Theorem \ref{MosTh} holds for $\mathfrak{M,N}$ and $\mathfrak S$.
\end{thm}

\begin{exa}
Let $I=\omega$, and let ${\mathfrak S}$ be the structure ${\mathfrak S}_{\omega}=(S(\omega); \subseteq,  \ll)$ where $x \ll y$ if and only if $x$ and $y$ are two singleton sets, say $x=\{m\}$ and $y=\{n\}$, such that $m<n$. This structure is a FO version of the MSO theory of $(\omega;<)$, which was shown to be decidable by B\"uchi \cite{Buchi62}. Let $\mathfrak M=(\omega;+)$. Consider the structure ${\mathfrak N}=(\omega^I; P_1)$ where $P_1(f)$ holds in $\mathfrak N$ if and only if there exist $j \in \omega$ such that $f(i)=0$ for every $i>j$. The structure ${\mathfrak N}$ is a generalized power of $\mathfrak M$ with respect to $\mathfrak S$. Indeed it is easy to check that $P_1(f)$ holds in $\mathfrak N$ if and only if 
$$\mathfrak S \models \exists X_1 \forall X_2 (X_1 \ll X_2 \rightarrow X_2 \subseteq T)$$
where
$$T=\{ i \in \omega \ | \  {\mathfrak M} \models f(i)+f(i)=f(i) \} 
.$$ 
By Theorem \ref{fvthm}, the FO theory of $\mathfrak N$ is decidable, since the FO theories of $\mathfrak M$ and ${\mathfrak S}_{\omega}$ are decidable by \cite{Pre29,Buchi62}.
\end{exa}

Feferman and Vaught also define the notion of {\em generalized weak power} of a structure ${\mathfrak M}=(\Sigma;R^{\mathfrak M}_1,\dots,R^{\mathfrak M}_k,P_e^{\mathfrak M})$ with respect to some index structure $\mathfrak S$. This notion generalizes the one of weak power by allowing to deal with index structures of the form $\mathfrak S=(S^+(I);\subseteq,S_1,\dots,S_n)$ where the $S_i$'s denote relations over $S^+(I)$. Feferman and Vaught prove that  Theorem \ref{MosTh} still holds for generalized weak powers, with the same modifications as for direct weak powers.

Let us consider the case where ${\mathfrak S}$ is the structure ${\mathfrak S}_{<\omega}=(S^+(\omega);\subseteq,\ll)$. The FO theory of ${\mathfrak S}_{<\omega}$ is a variant of the weak MSO theory of $(\omega;<)$. In this case there is a close correspondence between relations which admit a reduction sequence with respect to $\mathfrak S$, and 
 $MSO({\mathfrak M})$-definable relations, or equivalently $\mathfrak M$-recognizable relations
(by Theorem \ref{logicauto}).

Consider indeed an alphabet $\Sigma$, and the application $\mu$ which maps every finite word $w$ over $\Sigma$ to the $\omega-$word $\mu(w)=w \#^\omega$ over $(\Sigma \cup \{\#\})^\omega$. The word $\mu(w)$ can be seen as an element of $(\Sigma \cup \{\#\})^{(\omega)}_{\#}$. Given an $n-$ary relation $R$ over $\Sigma^*$, we set $\mu(R)=\{(\mu(w_1),\dots,\mu(w_n)) | (w_1,\dots,w_n) \in R \}$. The relation $\mu(R)$ can be seen as a subset of the set of sequences $f:\omega \to \omega \cup \{\#\}$ such that $f(i)\ne \#$ for finitely many values of $i$, i.e. as a subset of $(\Sigma \cup \{\#\})^{(\omega)}_{\#}$.

\begin{prop}\label{prop:transfer}
Let ${\mathfrak M}=(\Sigma;\dots)$ be a   structure. For every $n\geq 1$ and every $n-$ary relation $R$ over $\Sigma^*$, the relation $R$ is $MSO(\mathfrak M)$-definable if and only if the relation $\mu(R)$ admits a reduction sequence with respect to ${\mathfrak M}_{\#}$ and ${\mathfrak S}_{<\omega}$. 
\end{prop}

We shall illustrate this proposition by some example, and leave the
proof to the reader.\bigskip

\noindent{\bf Example \ref{exa:msodef} revisited}.
Let $\mathfrak M=(\omega;+)$, and let $R$ denote the set of words $w$ over $\omega$ that contain only even symbols. The relation $R$ is $MSO({\mathfrak M})$-definable by the sentence $\forall y \  \alpha_F(y)$,  
where $F(x): \ \exists z (z+z = x)$.

In this case $\mu(R)$ corresponds to the set of sequences $f \in (\omega \cup \{\#\})^{(\omega)}_{\#}$ such that 
\begin{enumerate}[$\bullet$]
\item the set of indexes $i$ such that $f(i)\ne \#$ is an initial segment of $\omega$;
\item for every $i$ such that $f(i) \ne \#$, $f(i)$ is even.  
\end{enumerate}
This implies that $\mu(R)$ admits a reduction sequence with respect to ${\mathfrak M}_{\#}$ and ${\mathfrak S}_{<\omega}$. Indeed we have $f \in \mu(A)$ if and only if 
$${\mathfrak S}_{<\omega} \models  \forall X \forall Y ((Y \subseteq T \wedge X \ll Y) \rightarrow X \subseteq T) \wedge \forall Y \  (Y \subseteq T \rightarrow Y \subseteq T')$$ where
$$T=\{ i \in \omega \ | \  {\mathfrak M}_{\#} \models f(i) \ne \# \} 
$$ and
$$T'=\{ i \in \omega \ | \  {\mathfrak M}_{\#} \models \exists z (z+z=f(i)) \} 
$$
The second part of the above formula is a direct translation of the MSO sentence $\forall y \  \alpha_F(y)$ used to prove the $MSO(\mathfrak M)$-definability of $R$. It is not difficult to prove that a similar translation is possible for every $MSO({\mathfrak M})$-definable relation. 

Recall that by Proposition \ref{logicauto}, $MSO(\mathfrak M)$-definability and $\mathfrak M$-recognizability are equivalent. Proposition \ref{prop:transfer} allows to revisit our previous results in terms of powers of structures:
\begin{enumerate}[$\bullet$]
\item The closure of $\mathfrak M$-recognizable relations under boolean operations, projection and cylindrification (Proposition \ref{prop:cloture}), could have been proven as a consequence of Proposition \ref{prop:transfer} together with the closure under definability of relations which admit a reduction sequence.

\item  If ${\mathfrak N}=(N;R_1,\dots,R_n)$ is an ${\mathfrak M}$-automatic structure then by Proposition \ref{logicauto} and \ref{prop:transfer} the image by $\mu$ of the domain and base relations of ${\mathfrak N}$ admit a reduction sequence with respect to ${\mathfrak M}_{\#}$ and ${\mathfrak S}_{<\omega}$. Thus 
every ${\mathfrak M}$-automatic structure is isomorphic to some relativized generalized weak power of ${\mathfrak M}_{\#}$ with respect to ${\mathfrak S}_{<\omega}$. Moreover by Feferman-Vaught Theorem, $FO({\mathfrak N})$ reduces to $FO({\mathfrak M}_{\#})$ and $FO({\mathfrak S}_{< \omega})$. Now one can reduce $FO({\mathfrak M}_{\#})$ to  $FO({\mathfrak M})$, and $FO({\mathfrak S}_{< \omega})$ is decidable by \cite{Buchi60}. Finally this proves that $FO({\mathfrak N})$ reduces to $FO({\mathfrak M})$, that is, Theorem \ref{thm:hodgson2}.

\item In the same way as $\mathfrak M$-automata over finite words correspond to generalized weak powers with respect to the index structure $\mathfrak S_{<\omega}$, one can define a notion of $\mathfrak M$-automata over $\omega-$words and show that it corresponds to generalized powers with respect to the index structure $\mathfrak S_{\omega}$.
\end{enumerate}

\section{Applications}\label{sec:applications}

\subsection{An application to monadic chain logic over iteration}\label{iteration}

We apply the previous results to improve a recent result of Kuske and Lohrey \cite{KL06}. This application was brought to our attention by Wolfgang Thomas.

In \cite{KL06} the authors consider decidability issues related to monadic second-order chain logic, and applications to pushdown systems. Given a structure ${\mathfrak A}=(A;<,\dots)$ where $<$ denotes a partial ordering, the monadic second-order chain logic of ${\mathfrak A}$, which will be denoted by $MSO^{ch}({\mathfrak A})$, is the fragment of the MSO theory of $\mathfrak A$ where monadic second order quantifications are restricted to chains (i.e. linearly ordered subsets) with respect to $<$. The logic $MSO^{ch}$ was first investigated in \cite{Tho87}.

Consider a   ${\mathfrak L}$-structure ${\mathfrak M}=(\Sigma;R_1,\dots,R_m)$. The {\em basic iteration of ${\mathfrak M}$} is the structure 
$${\mathfrak M}^*_{ba}=(\Sigma^*;\preceq, \hat{{A}}_{R_1}, \dots, \hat{{A}}_{R_m})$$
where, for every relational symbol $R_j$ with arity $n$,
$$\hat{{A}}_{R_j}=\{(ua_1,\dots,ua_n) \ | \ u \in \Sigma^* , (a_1,\dots,a_n) \in R_j^{\mathfrak M} \}$$

Kuske and Lohrey prove the following result.

\begin{thm}\cite[Theorem 4.10]{KL06} For every   structure $\mathfrak M$, the $MSO^{ch}$ theory of ${\mathfrak M}^*_{ba}$ reduces to the $FO$ theory of $\mathfrak M$.
\end{thm}

We can improve this result by replacing the structure ${\mathfrak M}^*_{ba}$ with the structure ${\mathcal S}_{\mathfrak M}$ which we introduced in Section \ref{automatic}.

\begin{thm}
For every   structure ${\mathfrak M}=(\Sigma;R_1,\dots,R_m)$,  the $MSO^{ch}$ theory of 
${\mathcal S}_{\mathfrak M}=(\Sigma^*;EqLength,\preceq,{{A}_{R_1}}, \dots, {{A}_{R_m}},A_=)$ reduces to the $FO$ theory of ${\mathfrak M}$.
\end{thm}

Observe that  all predicates $\hat{A}_{R_i}$ can be defined in ${\mathcal S}_{\mathfrak M}$.

\begin{proof}
The proof consists in two main steps. The first one is to reduce the $MSO^{ch}$ theory of ${\mathcal S}_{\mathfrak M}$ to the $FO$ theory of some structure ${\mathcal S}'_{\mathfrak M}$ with domain $\Sigma^\omega$. The technique is an adaptation from \cite[Section 4]{Tho92}. The second step consists in proving that ${\mathcal S}'_{\mathfrak M}$ is a generalized power with respect to $\mathfrak M$ and ${\mathfrak S}_{\omega}$, which allows then to use Feferman-Vaught Theorem and the decidability of $FO({\mathfrak S}_{\omega})$ to conclude.

We first consider a variant of $MSO^{ch}$ of ${\mathcal S}_{\mathfrak M}$ where only second-order variables occur; this can be done by introducing the inclusion predicate $X_1 \subseteq X_2$, and replacing relations between elements by the corresponding relations between singleton sets.

Now every chain $X$ of elements of $\Sigma^*$ can be represented by a couple $(u,v)$ of elements of $\Sigma^\omega$ in the following way: 
\begin{enumerate}[$\bullet$]
\item $u$ corresponds to the ``direction" of the chain $X$, i.e. is such that all elements of $X$ are prefixes of $u$ (note that if $X$ is infinite then there exists a unique such $u$);
\item $v$ indicates which prefixes of $u$ belong to $X$, in the following way: for every integer $i$, we have $v[i]=u[i]$ if and only if the prefix of $u$ of length $i$ belongs to $X$.
\end{enumerate} 
According to this definition, any couple $(u,v)$ of elements of $\Sigma^\omega$ represents a single chain which will be denoted by $ch(u,v)$. 

The previous encoding allows to reduce the $MSO^{ch}$ theory of ${\mathcal S}_{\mathfrak M}$ to the FO theory of the structure:
$${\mathcal S}'_{\mathfrak M}=(\Sigma^\omega;\approx,\subseteq',EqLength',\preceq',{A}'_{R_1},\dots, A'_{R_n})$$
where
\begin{enumerate}[$\bullet$]
\item $\approx(u_1,v_1,u_2,v_2)$ holds if and only if $ch(u_1,v_1) = ch(u_2,v_2)$;
\item $\subseteq'(u_1,v_1,u_2,v_2)$ holds if and only if $ch(u_1,v_1) \subseteq ch(u_2,v_2)$;
\item $EqLength'(u_1,v_1,u_2,v_2)$ holds if and only if there exist two words $u,u' \in \Sigma^*$ such that $ch(u_1,v_1)=\{u\}$, $ch(u_2,v_2)=\{u'\}$, and ${\mathcal S}_{\mathfrak M} \models EqLength(u,u')$ ;
\item $\preceq'(u_1,v_1,u_2,v_2)$ holds if and only if there exist two words $u,u' \in \Sigma^*$ such that $ch(u_1,v_1)=\{u\}$, $ch(u_2,v_2)=\{u'\}$, and ${\mathcal S}_{\mathfrak M} \models u \preceq u'$;
\item For every $i$, if $R_i$ is a $n-$ary relation, then $A'_{R_i}(u_1,v_1,u_2,v_2,\dots,u_n,v_n)$ holds if and only if there exist words $w_1,\dots,w_n \in \Sigma^*$ such that $ch(u_j,v_j)=\{w_j\}$ for every $j$,  and ${\mathcal S}_{\mathfrak M} \models A_{R_i}(w_1,\dots,w_n)$.
\item $A'_{=}(u_1,v_1,u_2,v_2)$ holds if and only if there exist two words $u,u' \in \Sigma^*$ such that $ch(u_1,v_1)=\{u\}$, $ch(u_2,v_2)=\{u'\}$, and ${\mathcal S}_{\mathfrak M} \models A_=(u,u')$;
\end{enumerate}

We shall prove that ${\mathcal S}'_{\mathfrak M}$ is a generalized power of $\mathfrak M$ with respect to ${\mathfrak S}_{\omega}$. To this aim, let us prove that all base relations of ${\mathcal S}'_{\mathfrak M}$ admit a reduction sequence with respect to $\mathfrak M$ and ${\mathfrak S}_{\omega}$:
\begin{enumerate}[$\bullet$]
\item $\subseteq'(u_1,v_1,u_2,v_2)$ holds if and only if
$${\mathfrak S}_{\omega} \models T_1 \subseteq T_2 \wedge \exists X (IS(X,T_3) \wedge T_1 \subseteq X)$$
where 
$$T_1=\{i \in \omega \  | \ u_1(i)=v_1(i) \},$$
$$T_2=\{i \in \omega \  | \ u_2(i)=v_2(i) \},$$
$$T_3=\{i \in \omega \  | \ u_1(i)=u_2(i) \},$$
and $IS(X,T_3)$ is a formula which expresses that $X$ is the greatest initial segment of $\omega$ which is contained in $T_3$ (this is expressible in ${\mathfrak S}_{\omega}$);
\item $\approx(u_1,v_1,u_2,v_2)$ holds if and only if both $\subseteq'(u_1,v_1,u_2,v_2)$ and $\subseteq'(u_2,v_2,u_1,v_1)$ hold, from which we can deduce a reduction sequence for the relation $\approx$;
\item $EqLength'(u_1,v_1,u_2,v_2)$ holds if and only if $T_1$ and $T_2$ are singleton sets and $T_1=T_2$, i.e. if 
$${\mathfrak S}_{\omega} \models \exists Y (T_1 \ll Y) \wedge T_1=T_2$$
(with the same notations as above);
\item $\preceq'(u_1,v_1,u_2,v_2)$ holds if and only if 
$${\mathfrak S}_{\omega} \models \exists Y (T_1 \ll Y) \wedge (T_1=T_2 \vee T_1 \ll T_2) \wedge \exists X (IS(X,T_3) \wedge T_1 \subseteq X);$$
\item for every $i$, if $R_i$ is a $n-$ary relation, then $A'_{R_i}(u_1,v_1,u_2,v_2,\dots,u_n,v_n)$ holds if and only if all sets $U_j=\{i \in \omega \  | \ u_j(i)=v_j(i) \}$, $j=1,2,\dots,m$, are singleton sets and are equal, and are included in the set 
$$U=\{i \in \omega \  | \ {\mathfrak M} \models R_j(u_1(i),u_2(i),\dots,u_n(i)) \}.$$
These properties can be expressed in ${\mathfrak S}_{\omega}$.

\item the case of $A'_{=}(u_1,v_1,u_2,v_2)$ is similar to the previous case, with $n=2$ and $=$ in place of $R_j$;

\end{enumerate}

We have proved that ${\mathcal S}'_{\mathfrak M}$ is a generalized power of $\mathfrak M$ with respect to ${\mathfrak S}_{\omega}$. By Theorem \ref{fvthm}, $FO({\mathcal S}'_{\mathfrak M})$ reduces to the FO theories of ${\mathfrak M}$ and ${\mathfrak S}_{\omega}$. Now $FO({\mathfrak S}_{\omega})$ is decidable by B\"uchi  \cite{Buchi62}, thus $FO({\mathcal S}'_{\mathfrak M})$ reduces to $FO({\mathfrak M})$.\end{proof}

\subsection{Ordinal addition and $(\omega;+)$-recognizability}\label{subsection:ordinal}

We shall focus now on the case ${\mathfrak M}=(\omega;+)$, where $+$ denotes the graph of addition. In this case we present another logical characterization of ${\mathfrak M}$-recognizable relations in terms of ordinal theories. This is essentially a reformulation of known results.

In the sequel we consider structures of the form $(\alpha;+)$ where $\alpha$ is an ordinal. The domain is the set of ordinals less than $\alpha$, and $+$ is interpreted as the graph of ordinal addition restricted to the domain.

Feferman and Vaught prove in \cite{FV59} that for every ordinal $\gamma$ the 
structure $(\omega^\gamma;+)$ is isomorphic to some generalized weak power of $(\omega;+)$ with respect to
$(S^+(\gamma); \subseteq,  \ll)$\footnote{As a corollary, the FO theory of $(\omega^\gamma;+)$ reduces to the FO theory of $(\omega;+)$ (Presburger Arithmetic, which is decidable \cite{Pre29}) and the weak MSO theory of $(\gamma,<)$. The latter was proved to be decidable by B\"uchi \cite{Buchi65} a few years after Feferman-Vaught' work, which implies the decidability of the FO theory of $(\omega^\gamma;+)$. }.

In particular for $\gamma=\omega$ their result, combined with B\"uchi's result, implies that via some encoding all relations definable in $(\omega^\omega;+)$ are $(\omega;+)$-recognizable, and that the theory of $(\omega^\omega;+)$ is decidable.

Let us be more specific. We first recall some useful results on ordinal arithmetic; all of them can be found e.g. in Sierpinski's book \cite[chap.XIV]{Sie65}

\begin{prop} [Cantor normal form for ordinals]\label{cantor}
Every ordinal $\alpha>0$ can be written uniquely as
\[ 
\alpha= \omega^{\alpha_1} a_1 
+ \cdots + \omega^{\alpha_k} a_k
\]
where $\alpha_1,\alpha_{2},\dots,\alpha_k$ is a decreasing sequence of ordinals, and $0< a_i < \omega$.
\end{prop}

The following proposition relates the Cantor normal form of the ordinal $\alpha+\beta$ to the one of $\alpha$ and $\beta$.

\begin{prop}\label{addition}
Let $\alpha= \omega^{\alpha_1} a_1 
+ \cdots + \omega^{\alpha_k} a_k$ and 
$ 
\beta= \omega^{\beta_1} b_1 
+ \cdots+  \omega^{\beta_l} b_l$ be two ordinals $>0$ in Cantor normal form. 
\begin{enumerate}[$\bullet$]
\item If $\alpha_1<\beta_1$ then $\alpha+\beta=\beta$
\item If $\alpha_1 \geq \beta_1$ and if $\alpha_j=\beta_1$ for some $j$, then 
\[ \alpha+\beta= 
(\omega^{\alpha_1} a_1 + \cdots + \omega^{\alpha_{j-1}} a_{j-1}) 
+ \omega^{\alpha_{j}}(a_j+b_1) 
+ (\omega^{\beta_2} b_2 + \cdots+  \omega^{\beta_l} b_l) 
\]
\item If $\alpha_1 \geq \beta_1$ and if $\alpha_j \ne \beta_1$ 
for every $j$, then 
\[ \alpha+\beta= 
(\omega^{\alpha_1} a_1 + \cdots + \omega^{\alpha_{m}} a_{m}) 
+ (\omega^{\beta_1} b_1 + \cdots+  \omega^{\beta_l} b_l) 
\]
where $m$ is the greatest index for which $\alpha_m > \beta_1$.
\end{enumerate}

\end{prop}

Consider now the function $f: \omega^\omega  \to \omega^*$ which maps every ordinal $\alpha<\omega^\omega$, written in Cantor normal form as $\alpha=\sum_{i=m}^{i=0}\omega^i a_i$ with $a_i<\omega$ and $a_m \ne 0$, to the word $c(\alpha)=a_0 \dots a_m$ over the alphabet $\omega$.  Given $n$ ordinals $\alpha_1,\dots,\alpha_n$, we define $c(\alpha_1,\dots,\alpha_n)$ as $\langle c(\alpha_1),\dots,c(\alpha_n) \rangle$, where we choose $0$ as the padding symbol $\#$.

\begin{exa}\label{ex:additionordinale}
Consider the ordinals 
\[ \alpha= \omega^{6} \cdot 5+ \omega^{4} \cdot 4 + \omega^3 \cdot 3 + \omega^1 \cdot 2 + \omega^0 \cdot 11,\] and 
\[\beta=  \omega^3 \cdot 17 + \omega^2 \cdot 6 + \omega^1 \cdot 2.\] 
Then by Proposition \ref{addition} (second case), the ordinal $\gamma=\alpha+\beta$ equals 
\[\gamma= (\omega^{6} \cdot 5+ \omega^{4} \cdot 4) + \omega^3 \cdot (3+17) + (\omega^2 \cdot 6 + \omega^1 \cdot 2).\]
We have 
\[
c(\alpha,\beta,\gamma)=
%%%
\left( 
\begin{array}{c} 
11 \\ 
0 \\ 
0 \\ 
\end{array} 
\right)
%%%
\left( 
\begin{array}{c} 
2 \\ 
2 \\ 
2 \\ 
\end{array} 
\right)
%%%
%%%
\left( 
\begin{array}{c} 
0 \\ 
6 \\ 
6 \\ 
\end{array} 
\right)
%%%
\left( 
\begin{array}{c} 
3 \\ 
17 \\ 
20 \\ 
\end{array} 
\right)
%%%
\left( 
\begin{array}{c} 
4 \\ 
0 \\ 
4 \\ 
\end{array} 
\right)
%%%
\left( 
\begin{array}{c} 
0 \\ 
0 \\ 
0 \\ 
\end{array} 
\right)
%%%%%%%% % %
\left( 
\begin{array}{c} 
5 \\ 
0 \\ 
5 \\ 
\end{array} 
\right).
\]

\end{exa}

\

Proposition \ref{addition} is the key argument in Feferman-Vaught' proof that $(\omega^\gamma;+)$ is isomorphic to some generalized weak power of $(\omega;+)$ with respect to
$(S^+(\gamma); \subseteq,  \ll)$. Let us reformulate their ideas in terms of $(\omega;+)$-automata.

\begin{prop}
The image by $c$ of the graph of addition for ordinals $<\omega^\omega$ is $(\omega;+)$-recognizable.
\end{prop}

\begin{proof} \ A convenient $(\omega;+)$-automaton which recognizes the language 
\[X~=~\{ c(\alpha,\beta,\gamma) \ | \ \alpha,\beta,\gamma < \omega^\omega, \ \alpha+\beta=\gamma \}\]
 is pictured in Figure
  \ref{fig:autoaddition}, where 
\begin{enumerate}[$\bullet$]
\item $\varphi_1(x,y,z):  z=y$
\item $\varphi_2(x,y,z): \ y \ne 0 \wedge z=x+y$
\item $\varphi_3(x,y,z): \ y = 0 \wedge z=x$
\end{enumerate}

\begin{figure}[htbp]
  \begin{center}
    \gasset{Nw=6,Nh=6,loopdiam=7}
    \begin{picture}(35,20)(0,-5)         
    \node[Nmarks=i,iangle=270](0)(10,5){$q_0$}    
    \node[Nmarks=fi,fangle=0,iangle=270](1)(25,5){$q_1$}
    \drawloop(0){$\varphi_1$} 
    \drawedge(0,1){$\varphi_2$}
    \drawloop(1){$\varphi_3$}
    \end{picture}
    \caption{An $(\omega;+)$-automaton for ordinal addition}
    \label{fig:autoaddition}
  \end{center}
\end{figure}
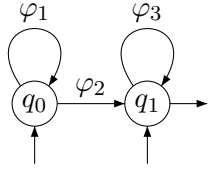
  
This automaton has two states $q_0,q_1$. Both are initial states, and only $q_1$ is final. Using $q_1$ as the initial state allows to deal with the case $\beta=0$; in this case we have $\gamma=\alpha$, which is checked by the transition labelled by $\varphi_3$.

Using $q_0$ as an initial state allows to deal with the case $\beta \ne 0$. In this case let $\omega^{\beta_1}$ denote the greatest power of $\omega$ which appears in the Cantor normal form of $\beta$. The transition labelled by $\varphi_1$ allows to deal with coefficients of powers $\omega^i$ where $i < \beta_1$; for these powers the corresponding coefficients of $\beta$ and $\alpha+\beta$ must be equal. The transition labelled by $\varphi_2$ corresponds to the power $\omega^{\beta_1}$. Then for all powers $\omega^j$ such that $j > \beta_1$, the corresponding coefficients of $\alpha$ and $\alpha+\beta$ coincide; this corresponds to the transition labelled by $\varphi_3$.
\end{proof}

%%%

We can provide now a characterization of ${\mathfrak M}$-recognizable relations for the case ${\mathfrak M}=(\omega;+)$.

\begin{prop}
For every $n\geq 1$, and every $n-$ary relation $R$ over $\omega^\omega$, the relation $R$ is definable in $(\omega^\omega;+)$ if and only if  $c(R)$ is $(\omega;+)$-recognizable.
\end{prop}

\begin{proof} (sketch)
The ``only if" part comes from the fact that the range of $c$, as well as the graph of ordinal addition, are $(\omega;+)$-recognizable. Thus $(\omega^\omega;+)$ is $(\omega;+)$-automatic, and the result follows from Theorem \ref{thm:hodgson2}.

For the converse one can use again B\"uchi's encoding technique as in Theorem \ref{caractEES}. Assume that $c(R)$ is $(\omega;+)$-recognizable by some $(\omega;+)$-automaton $\mathcal A$ whose set of states is $Q=\{q_0,q_1,\dots,q_m\}$. We can define $R$ in $(\omega^\omega;+)$ by a formula $\varphi(\alpha_1,\dots,\alpha_n)$ which expresses the existence of a successful path of $\mathcal A$, say $(q_{j_0},q_{j_2},\dots,q_{j_{m}})$, labelled by $c(\alpha_1,\dots,\alpha_n)$. The formula encodes the path with an ordinal of the form $\gamma=\omega^m j_m + \omega^{m-1} j_{m-1} + \dots +\omega^0 j_0$. 

We need to define the following auxiliary predicates (we explain briefly how to define them in $(\omega^\omega;+)$):

\begin{enumerate}[$\bullet$]
\item $\alpha < \beta$ (we have $\alpha<\beta$ if and only if there exists some non-null ordinal $\gamma$ such that $\beta=\alpha+\gamma$);
\item the function $(x_1,\dots,x_n) \mapsto \max(x_1,\dots,x_n)$;
\item ``to be a limit ordinal less than $\omega^\omega$" (these are non-null ordinals which have no predecessor with respect to $<$);
\item $Pow(x)$ which holds iff $x$ is a power of $\omega$ less than $\omega^\omega$ (which holds iff $x$ is a limit ordinal and there do not exist limit ordinals $\beta,\gamma$ such that $x=\beta+\gamma$ and $\gamma \leq \beta$);
\item For every $i<\omega$, the relation $Mult_i(x)$ which holds iff $x$ an ordinal of the form $\omega^k \cdot i$ (easily definable with the predicate $Pow(x)$);
\item the function $x \mapsto x \omega$ (for $x \ne 0$, the ordinal $x\omega$ is the least power of $\omega$ greater than $x$);
\item $App(x,y)$ which holds iff $y$ is a power of $\omega$ which appears in the Cantor normal form of $x$ (this holds if and only if $Pow(y)$ holds and moreover there exist ordinals  $\beta_1,\beta_2$ such that $x=\beta_1+y+\beta_2$ and $\beta_2<y$);
\item $AddCoef(x,y,z)$ which holds if and only if there exist $i,j,k \leq \omega$ such that $x= \omega^k i$, $y=\omega^k j$ and $z=\omega^k (i+j)$ -- which is equivalent to saying that $z=x+y$ and there exists exactly one ordinal $\alpha$ such that $App(x,\alpha)\wedge App(y,\alpha) \wedge App(z,\alpha)$ holds.

\item $Term(x,y,z)$ which holds iff $z$ is a power of $\omega$, say $z=\omega^k$, $y=\omega^k i$ for some $i<\omega$, and $y$ is the term which corresponds to $\omega^k$ in the Cantor normal form of $x$. The relation $Term(x,y,z)$ holds if and only if $Pow(z)$ holds, $z$ is the only power of $\omega$ which appears in the Cantor normal form of $y$, and there exist ordinals $\beta_1,\beta_2$ such that $x=\beta_1+y+\beta_2$ with $\beta_2<y$, and $z$ do not appear in the Cantor normal forms of $\beta_1$ and $\beta_2$. 
\item For every formula $\psi(x_1,\dots,x_n)$ in the language $\{+,=\}$ one can define the predicate $S_\psi(y_1,\dots,y_n,z)$ which holds if and only if $z$ is a power of $\omega$, say $z=\omega^k$, and if we denote by $a_1, \dots, a_n$ the coefficients of $\omega^k$ in the Cantor normal forms of $y_1,\dots,y_n$, respectively, then $(\omega;+)\models \psi(a_1,\dots,a_n)$. The predicates $S_\psi$ can be defined from $Term$ and $AddCoef$ by induction on the construction of $\psi$.
\item $Cod_i(x,y)$ expresses that $y$ is a power of $\omega$ and the coefficient of $y$ in the Cantor normal form of $x$ equals $i$. This predicate is easily definable from the predicates $Mult_i$ and $Term$.
\end{enumerate}

Finally we can define the formula $\varphi(\alpha_1,\dots,\alpha_n)$ as 

\begin{eqnarray}
\lefteqn{
\exists \gamma \bigg(\gamma < \max(\alpha_1,\dots,\alpha_n) \cdot \omega    \ 
}
\\  \label{eq:l2}
& & \wedge \bigvee_{q_i \in I} Cod_i(\gamma,\omega^0)   \\  \label{eq:l3}
& & \wedge \big( \forall \beta ((Pow(\beta) \wedge \beta \leq \max(\alpha_1,\dots,\alpha_n) \cdot \omega) \longrightarrow  \\ \label{eq:l4} 
& & \ \ \ \  \bigvee_{(q_i,\psi,q_j) \in E} (Cod_i(\gamma,\beta)\wedge  S_{\psi}(\alpha_1,\dots,\alpha_n,\beta) \wedge Cod_j(\gamma,\beta \omega)) \big)  \\ \label{eq:l5}
& & \wedge  \bigvee_{q_i \in T} Cod_i(\gamma, \max(\alpha_1,\dots,\alpha_n) \cdot \omega) \bigg) 
\end{eqnarray}
Line \ref{eq:l2} states that the first state of the sequence of states encoded by $\gamma$ is an initial state; lines \ref{eq:l3} and \ref{eq:l4} that consecutive states in the sequence use transitions of the automaton, and line \ref{eq:l5} that the last state of the sequence is terminal.
\end{proof}

\begin{rems}\hfill
\begin{enumerate}[$\bullet$]
\item
 One can prove that the graph of $x \mapsto \omega x$ is not ${\mathfrak M}$-recognizable, either in a direct way, or using the fact that  by \cite{Choffrut01} the theory of $(\omega^\omega;+,x \mapsto \omega x)$ is undecidable, while the theory of $(\omega^\omega;+,x \mapsto  x \omega)$ is decidable since the function $x \mapsto  x \omega$ is definable in $(\omega^\omega;+)$ which has a decidable theory.
\item We could reformulate the above results by replacing $(\omega^\omega;+)$ by the structure  
$~(\omega;\times,<_P)$, where $x <_P y$ holds if and only if $x<y$ and $x,y$ are prime numbers. In this case we encode every word $u=a_0\dots a_n$ over the alphabet $\omega$ by the integer $c'(u)=2^{a_0+1} 3^{a_1+1}\dots  p_n^{a_n+1}$ where $p_n$ denotes the $n-$th prime number. We refer to \cite{Mau97} for details about the link between $(\omega^\omega;+)$ and
 $(\omega;\times,<_P)$.
\end{enumerate}
\end{rems}
%%%%%%%%%%%%%%%%%%%%%%

%%%%%%%%%%%%%%%%%%%%%%%%

\section{An extension of the Feferman-Vaught formalism}\label{extFV}

The automata and logic that we introduced in the previous sections do not allow comparisons between symbols from different positions. For instance, for every structure ${\mathfrak M}$ whose domain is infinite, the language $\{ss \ | \ s \in |M|\}$ is not $MSO({\mathfrak M})$-definable (see Example \ref{exa:aa}). More generally, given any formula $\varphi(x,y)$ in the language of ${\mathfrak M}$, the language $\{s_1s_2 \ | \ M_\# \models \varphi(s_1,s_2)\}$ is not in general $MSO({\mathfrak M})$-definable.  

A natural way to add expressive power is to extend MSO with predicates such as $P(x,y)$ interpreted as ``$x,y$ are two positions in $w$ such that $w[x]=w[y]$", or more generally predicates interpreted as ``$x,y$ are two positions in $w$ such that ${\mathfrak M}_\# \models \varphi(w[x],w[y])$" (where $\varphi$ is some ${\mathfrak L}_\#$-formula).

However these extensions do not add expressive power when ${\mathfrak M}$ is finite, and lead to undecidable theories in case ${\mathfrak M}$ has an infinite domain (we refer the reader e.g. to  \cite{BDMSS05} where it is shown that much weaker related formalisms have undecidable FO theories).

Thus in order to get decidability results we have to restrict the use of these new predicates. Below we describe a syntactic fragment for which the satisfiability problem still reduces to the decidability of the first-order theory of ${\mathfrak M}$.

Given an ${\mathfrak L}$-structure ${\mathfrak M}=(\Sigma;\dots)$, we associate to every ${\mathfrak L}_\#-$formula $F$ with $m$ free variables some (new) $m-$ary relational symbol $\theta_F$. 

\begin{defi} 
We define $MSO^+({\mathfrak L})$ as MSO over the language $\{<,(\theta_F)_{F \in {\mathcal F}}\}$ where $\mathcal F$ denotes the set of ${\mathfrak L}_\#-$formulas with at least one free variable.
\end{defi}

The interpretation of $MSO^+({\mathfrak L})$ sentences is similar to $MSO({\mathfrak L})$, but for every ${\mathfrak L}_\#$-formula $F$ with $m$ free variables the interpretation of $\theta_F(x_1,\dots,x_m)$ is ``the positions $x_1,\dots,x_m$ in the word $w$ satisfy ${\mathfrak M}_\# \models F(w[x_1],\dots,w[x_m])$". 

\begin{defi}
We say that $X \subseteq \Sigma^*$ is $MSO^+({\mathfrak M})$-definable if there exists an $MSO^+({\mathfrak L})$-sentence $\varphi$ which defines $X$. The definition can be extended easily to the case of subsets $X \subseteq (\Sigma^*)^n$.
\end{defi}

Note that if one allows only $MSO^+({\mathfrak L})$ sentences where the predicates $\theta_F$ are unary, we get nothing but $MSO({\mathfrak L})$.

\begin{exa}\label{exa:extfv1}

Let ${\mathfrak M}=(\omega;+)$. 
\begin{enumerate}[$\bullet$]
\item The language $X \subseteq \omega^*$ of words $u$ over $\omega$ such that some symbol $s \in \Sigma$ appears at least twice in $u$ is $MSO^+({\mathfrak M})$-definable by the $MSO^+({\mathfrak L})$-sentence 
\[ \exists x  \exists y (x<y \wedge  \theta_F(x,y) ) \]
where
\[F(x_1,x_2): x_1 = x_2.\] 

\item The language $X' \subseteq \omega^*$ of words over $\omega$ of the form $u=s_0\dots s_m$ such that there exists $j \in \{0,\dots,m\}$ such that $s_k \geq 2 s_j$ whenever $k > j$, is $MSO^+({\mathfrak M})$-definable by the $MSO^+({\mathfrak L})$-sentence 
\[ \exists x ( \exists x' (x<x') \wedge \forall y (x < y \rightarrow \theta_G(x,y) ) \]
where
\[G(x_1,x_2): \exists z (x_2=x_1+x_1+z)\] denotes the formula which expresses that $x_2 \geq 2 x_1$.
\end{enumerate}
\end{exa}

\begin{exa}\label{exa:extfv2}
Let ${\mathfrak M}=(\Sigma^*;EqLength,\preceq, \{L_a\}_{a\in \Sigma})$ denote the EES structure $\mathcal S$ (see Section 
\ref{automatic}). 
The set of words $w=s_0 \dots s_m$ over the infinite alphabet $\Gamma=\Sigma^*$ such that all even positions carry the same symbol, and all odd positions carry a symbol which is a prefix of $s_0$, is $MSO^+({\mathfrak M})$-definable.  Indeed a convenient  $MSO^+({\mathfrak L})$-sentence is
{\setlength\arraycolsep{3pt}
\begin{eqnarray*}
\lefteqn{
\exists X [EvenPositions(X) \wedge \ 
}
\\
& & \wedge \exists x \in X \ \forall y \in X \ \theta_{F_1}(x,y) \wedge \exists z (\forall t \ \neg t<z \wedge \forall y \not\in X \ \theta_{F_2}(y,z))]
\end{eqnarray*}
}
where $EvenPositions(X)$ is an MSO-formula which expresses that $X$ consists in the set of even positions of $w$, and 
\[F_1(v_1,v_2): \ v_1=v_2 \ ; \] 
\[F_2(v_1,v_2): \ v_1 \preceq v_2.\]
\end{exa}

\bigskip

The formalism $MSO^+({\mathfrak L})$ is in general too expressive with respect to decidability, thus we have to consider a syntactic fragment of it. 

\begin{defi}
We define $MSO^+_R({\mathfrak L})$ as the syntactic fragment of $MSO^+({\mathfrak L})$ consisting in formulas of the form 
\[ \exists x_1 \dots \exists x_n \ \varphi(x_1,\dots,x_n) \]
where $\varphi$ is an $MSO^+({\mathfrak L})$-formula which satisfies the following constraint, which we denote by $(*)$:  {\em all predicates of the form $\theta_F$ in $\varphi$ have the form $\theta_F(x_1,\dots,x_n,y)$, i.e. contain at most one free variable distinct from the $x_i's$.}
\end{defi}

Note that formulas considered in Examples \ref{exa:extfv1} and \ref{exa:extfv2} are $MSO_R^+({\mathfrak L})$-formulas.

\begin{thm}\label{extensionFV}
The emptiness problem for $MSO_{R}^+({\mathfrak M})$-definable languages reduces to the decidability of the FO theory of ${\mathfrak M}$. 
\end{thm}

\begin{proof}
Let $\Sigma$ be the domain of $\mathfrak M$. To each $MSO^+_R({\mathfrak L})$-sentence $\psi$ of the form
\[ \exists x_1 \dots \exists x_n \ \varphi(x_1,\dots,x_n) \]
where $\varphi$ satisfies $(*)$ we associate in an effective way an
$MSO({\mathfrak L}')$-formula $\psi'$ where 
${\mathfrak L}'$ is obtained by adding to $\mathfrak L$ new constant symbols $c_1,\dots,c_n$, in order that for every ${\mathfrak L}$-structure ${\mathfrak M}$, the sentence $\psi$ is satisfiable by some word over $\Sigma$ if and only if there exists some ${\mathfrak L}'-$expansion ${\mathfrak M}'$ of ${\mathfrak M}$ such that the the set of words over $\Sigma$ defined by $\psi'$ is not empty.

The transformation proceeds as follows.
First, we can assume that all formulas of the form $\theta_F(x_1,\dots,x_n,y)$ which appear in $\varphi$ are such that $y$ appears freely in $\theta_F$: indeed if $y$ does not appear in $\theta_F$ then $\theta_F(x_1,\dots,x_n)$ is equivalent to $\exists y(y=x_n \wedge \theta_{F'}(x_1,\dots,x_{n-1},y))$ where $F'$ is obtained from $F$ by substituting $y$ for $x_n$. 

We define the $MSO({\mathfrak L}')$-formula $\psi'$ as 
\[ \exists x_1 \dots \exists x_n \ (\bigwedge_{i=1}^n \alpha_{F_i}(x_i) \wedge \varphi'(x_1,\dots,x_n)) \] 
where $F_i(y)$ denotes the formula $y=c_i$ and $\varphi'$ is obtained from $\varphi$ by replacing every formula $\theta_F(x_1,\dots,x_n,y)$ by the formula $\alpha_{F'}(y)$ where $F'$ is obtained from $F$ by replacing every occurence of $x_i$ by the constant symbol $c_i$.

It is easy to check that for every ${\mathfrak L}$-structure ${\mathfrak M}$, $\psi$ is satisfiable by some word model over $\Sigma$ if and only if there exists some ${\mathfrak L}'-$expansion ${\mathfrak M}'$ of ${\mathfrak M}$ such that $L(\psi') \ne \emptyset$.

The formula $\psi'$ involves the predicates $\alpha_{F_1}, \dots, \alpha_{F_n}$, and also predicates of the form $\alpha_F$ which appear in $\varphi'$, say $\alpha_{F_{n+1}}, \dots, \alpha_{F_p}$. By Proposition \ref{logicauto}, given ${\mathfrak M}'$ the question of whether the language defined by $\psi'$ is empty reduces to decide emptiness for the corresponding ${\mathfrak M}'-$automaton. This amounts to compute the set $E_{{\mathfrak M}'}$ of subsets $I\subseteq \{1,\dots,p\}$ such that there exists $a \in \Sigma$ such that (${\mathfrak M}' \models F_i(a)$ if and only if $i \in I$). Thus it suffices to compute all possible sets $E_{{\mathfrak M}'}$ for all ${\mathfrak L}'-$expansions ${\mathfrak M}'$ of ${\mathfrak M}$. This can be done effectively since for every subset $E$ of subsets of $\{1,\dots,p\}$, one can find an ${\mathfrak L}$-sentence $H_E$ such that ${\mathfrak M} \models H_E$ if and only if there exists some ${\mathfrak L}'-$expansion ${\mathfrak M}'$ of ${\mathfrak M}$ such that $E_{{\mathfrak M}'}=E$. Therefore we reduced our initial problem to the question of whether $\mathfrak M$ satisfies some sentence.
\end{proof}

\section{Discussion and conclusion}

The proof of Theorem \ref{extensionFV} makes uses of B\"uchi's decidability result for the WMSO theory of $(\omega;<)$. However the arguments are sufficiently general to apply to any decidable extension of WMSO. An interesting example is the WMSO theory $T_{card}$ of $\omega$, without $<$, but with the predicate $X \sim Y$ interpreted as ``$X$ and $Y$ have the same cardinality". This theory was proven to be decidable by Feferman and Vaught in \cite{FV59} by reduction to Presburger Arithmetic (by elimination of quantifiers, and without using the composition technique). For recent applications of this decidability result we refer the reader to the  papers \cite{KNR05,Rev05,Lugiez05}.

One can show that Theorem \ref{extensionFV} holds with $T_{card}$, which provides a class of theories which are both decidable and quite expressive. As an example, if we set ${\mathfrak M}=(\Sigma^*;EqLength,\preceq, \{L_a\}_{a\in \Sigma})$ (the EES structure, whose FO theory is decidable \cite{EES69}), then the corresponding syntactic fragment allows to express properties related to finite words $w$ over the alphabet $\Sigma'=\Sigma^*$ (that is, finite sequences of words over $\Sigma$) such as ``there exist two distinct symbols $s,s'$ appearing in $w$ such that at least one third of the symbols in $w$ are prefix of $s$, or have the same length as $s'$". Another interesting example is the case ${\mathfrak M}=(\omega;+)$. In this case we obtain a decidable fragment for words over the alphabet $\omega$, i.e. lists of natural numbers. This fragment might be an interesting formalism for the verification of programs which manipulate pointers and linked data structures.

By Proposition \ref{logicauto}, ${\mathfrak M}$-automata capture the logic $MSO({\mathfrak L})$. Thus a natural issue is to get an automata counterpart for the logic $MSO^+_R({\mathfrak L})$. An idea is to consider ${\mathfrak M}$-automata equipped with a finite number of ``write once" registers. In addition to the usual transitions of ${\mathfrak M}$-automata, these automata are allowed to write the current symbol in some empty register, and test whether the symbols currently stored in the registers and the current symbol satisfy some ${\mathfrak L}_\#-$sentence in ${\mathfrak M}_\#$. Once a symbol is stored in some register, the automaton cannot store any other symbol in this register. In order to capture the fragment $MSO^+_R({\mathfrak L})$, it seems that one should also allow non-deterministic $\epsilon-$transitions where the automaton chooses to store some symbol from the input alphabet in some (empty) register.

Another interesting issue would be to find (more natural) extensions of the Feferman-Vaught formalism  in the spirit of Theorem \ref{extensionFV}. The formalism $MSO^+({\mathfrak L})$ allows the use of predicates $\theta_F$ for all ${\mathfrak L}_\#-$formulas $F$, which makes necessary to consider the fragment $MSO^+_R({\mathfrak L})$ in order to get decidability results. It would be interesting to find other fragments of $MSO^+({\mathfrak L})$ obtained by imposing conditions on the ${\mathfrak L}_\#-$formulas $F$. One can consider e.g the case where we allow only the use of formulas $F$ which define equivalence relations in ${\mathfrak M}$. Note that similar results are already proven in the papers \cite{BDMSS05,Bouyer2002}.

Finally, it seems that all results in this paper can be extended rather easily to the case of infinite words as well as (in)finite binary trees, by relying on classical decidability results for MSO theories. 

\section*{Acknowledgements}

I wish to thank Wolfgang Thomas for his careful reading of a preliminary version of the paper, and for many important suggestions and corrections. I also thank Lev Beklemishev, Christian Choffrut, Serge Grigorieff, Anca Muscholl, Peter Revesz and Luc S\'egoufin for interesting discussions and comments. Finally I am grateful to the two referees for helpful suggestions to improve the quality of the paper.

\bibliographystyle{plain}

\bibliography{fv}

\end{document}